\documentclass[english,latin1,11pt]{article}
\usepackage[latin1]{inputenc}
\usepackage{float}
\usepackage{amsmath}
\usepackage{graphicx}
\usepackage{graphics,epsfig,natbib,}
\makeatletter

\usepackage{amscd,amsmath,amssymb,amsfonts,latexsym,mathrsfs}
\usepackage{amsthm}
\usepackage{enumerate}
\usepackage{rotate}

\usepackage{latexsym}

\usepackage{babel}

\makeatother

\begin{document}
\newtheorem{theo}{Theorem}
\newtheorem{defi}{Definition}
\newtheorem{coro}{Corollary}
\newtheorem{prop}{Proposition}[section]
\newtheorem{rem}{Remark}
\newtheorem{lemma}{Lemma}
\newtheorem{proper}{Property}[section]
\newtheorem{algo}{Algorithm}

\newcommand{\tk}{P^{n} (x, dy)}
\newcommand{\bm}{\boldmath}
\newcommand{\um}{\unboldmath}
\newcommand{\botheta}{\mbox{$\theta$}}
\newcommand{\boTheta}{\mbox{$\Theta$}}
\newcommand{\G}{\mbox{${\cal{G}}$}}
\newcommand{\M}{\mbox{${\cal{M}}$}}
\newcommand{\rcorr}{\mbox{corr}}
\newcommand{\rcov}{\mbox{cov}}
\newcommand{\RQMC}{\rm RQMC}
\newcommand{\rv}{\mbox{Var}}
\newcommand{\V}{\mbox{Var}}
\newcommand{\bolam}{\mbox{$\lambda$}}
\newcommand{\bopi}{\mbox{$p$}}
\newcommand{\ru}{{\mbox{Uniform}}}
\newcommand{\cas}{\mbox{${\cal{S}}$}}
\newcommand{\caf}{\mbox{${\cal{F}}$}}
\newcommand{\cag}{\mbox{${\cal{G}}$}}
\newcommand{\cac}{\mbox{${\cal{C}}$}}
\newcommand{\cau}{\mbox{${\cal{U}}$}}
\newcommand{\cae}{\mbox{${\cal{E}}$}}
\newcommand{\caz}{\mbox{${\cal{Z}}$}}
\newcommand{\cay}{\mbox{${\cal{Y}}$}}

\newcommand{\beq}{\begin{equation}}
\newcommand{\eeq}{\end{equation}}
\newcommand{\beqn}{\begin{eqnarray}}
\newcommand{\eeqn}{\end{eqnarray}}
\newcommand{\one}{{\mathbf 1}}
\newcommand{\RR}{{\mathbf R}}

\newcommand{\LL}{\mbox{${\cal{L}}$}}
\newcommand{\xe}{\mbox{E}}
\newcommand{\bu}{\mathbf{u}}
\newcommand{\bU}{\mathbf{U}}
\newcommand{\bW}{\mathbf{W}}
\newcommand{\bv}{\mathbf{v}}
\newcommand{\bw}{\mathbf{w}}
\newcommand{\bi}{\mathbf{i}}

\newcounter{example}[section]
\def\theexample{\thesection.\arabic{example}}
\setcounter{example}{0}

\renewenvironment{proof} {\noindent {\bf Proof } \it \bigskip}

\title{Interacting Multiple Try Algorithms with Different Proposal Distributions}

\author{Roberto Casarin\setcounter{footnote}{1}\footnotemark{} \setcounter{footnote}{4}\footnotemark{}\hspace{15pt} %
Radu V. Craiu\setcounter{footnote}{2}\footnotemark{}\hspace{15pt} %
Fabrizio Leisen\setcounter{footnote}{3}\footnotemark{}\hspace{15pt} \\\\%
{\centering \small\setcounter{footnote}{4}\footnotemark{}$\,$ Advanced School of Economics, Venice}\\%
{\centering \small\setcounter{footnote}{1}\footnotemark{} University of Brescia}\\%
{\centering \small\setcounter{footnote}{3}\footnotemark{} Universidad Carlos III de Madrid}\\%
{\centering \small\setcounter{footnote}{2}\footnotemark{} University of Toronto}%
}

\maketitle

\begin{abstract}
We propose a new class of interacting Markov chain Monte Carlo (MCMC) algorithms designed for increasing the efficiency of a modified  multiple-try Metropolis (MTM) algorithm. The extension with respect to the existing MCMC literature is twofold. The sampler proposed extends the basic MTM algorithm by allowing different proposal distributions in the multiple-try generation step. We exploit the structure of the MTM algorithm with different proposal distributions to naturally introduce an interacting MTM mechanism (IMTM) that expands the class of  population Monte Carlo methods. We show the validity of the algorithm and discuss the choice of the selection weights and of the different proposals. We provide numerical studies which show that the new algorithm can perform better than the basic MTM algorithm and that the interaction mechanism allows the IMTM to efficiently explore the state space.
\end{abstract}

\section{Introduction}

Markov chain Monte Carlo (MCMC) algorithms are now essential for the analysis of complex statistical models. In the MCMC universe, one of the most widely used class of algorithms is defined by the Metropolis-Hastings (MH)   \citep{metrop,Hast:70:MCS} and its variants.  An important generalization of the standard MH formulation is represented by the multiple-try Metropolis (MTM)   \citep{Liu:2000cr}.  While in the MH formulation  one accepts or rejects a single proposed move, the MTM is designed so that the next state of the chain is selected among  multiple  proposals. The multiple-proposal setup can be used effectively to explore the sample space of the target distribution and subsequent developments made use of this added flexibility.  For instance,  \cite{Craiu:2007rr} propose to use antithetic and quasi-Monte Carlo samples to generate the proposals and to improve the efficiency of the algorithm while \cite{Pandolfi:2010a} and \cite{Pandolfi:2010b} apply the multiple-proposal idea to a trans-dimensional setup and combine Reversible Jump MCMC with MTM.

This work further generalizes  the MTM algorithm presented in \cite{Liu:2000cr} in two directions.   First, we show that the original MTM transition kernel  can be modified to allow for different proposal  distributions in the multiple-try generation step while preserving  the ergodicity of the chain. The use of different proposal distributions gives more freedom in designing MTM algorithms for target distributions that require different proposals across the sample space.  An important challenge remains the choice of the  distributions used to generate the proposals and we propose to address it by expanding upon methods used   within the population Monte Carlo class of algorithms.
 
The class of population Monte Carlo procedures \citep{pop-mcmc,DelMoral:Miclo:00,DelMoral:04,pop-static} has been designed to address the inefficiency of classical MCMC samplers in complex applications involving multimodal and high dimensional target distributions \citep{Pritchard:2000fk,hhs}.
 Its formulation  relies on a number of MCMC processes that are run in parallel while learning from one another about the geography of the target distribution. 

A second contribution of the paper is  finding reliable generic methods for constructing the proposal distributions for the MTM algorithm.We propose an interacting MCMC sampling design for the MTM that preserves the Markovian property. More specifically, in the proposed interacting MTM (IMTM) algorithm, we allow the distinct proposal distributions to use information produced by a population of auxiliary chains. We infer that the resulting performance of the MTM is tightly connected to the performance of the chains' population. In order to maximize the latter, we propose a number of strategies that can be used to tune  the auxiliary chains.
We also adapt previous extensions of the MTM and link the use of  stochastic overrelaxation,  random-ray Monte Carlo method \citep[see][]{Liu:2000cr} and simulated annealing to IMTM.

In the next section we discuss the IMTM algorithm, propose a number of alternative implementations and prove their ergodicity. In Section 3 we focus on some special cases of the IMTM algorithm and in Section 4 the performance of the methods proposed is demonstrated with simulations and real examples. We end the paper with a discussion of future directions for research.

\section{Interacting Monte Carlo Chains for MTM}

 We begin by describing the MTM and its extension for using different proposal distributions.
 
\subsection{Multiple-Try Metropolis With Different Proposal Distributions}

Suppose that of interest is sampling from a distribution $\pi$ that has support in
$\mathcal{Y}\subset\RR^d$ and is known up
to a normalizing constant.   Assuming that the current state of the chain is $x$, the update defined by the  MTM algorithm of \cite{Liu:2000cr}  is described in Algorithm 
\ref{alg1}.

\begin{center}
\begin{minipage}[thp]{330pt} \par\hrule\vspace{5pt}
\begin{algo}{Multiple-try Metropolis Algorithm (MTM)}
\par\vspace{5pt}\hrule
\begin{enumerate}
\item  Draw $M$ trial proposals $y_1,\ldots,y_M$ from the proposal
distribution $T(\cdot|x)$. Compute $w(y_j,x)$ for each $j\in\{1,\ldots,M\}$, where $w(y,x)=\pi(y) T(x|y)\lambda(y,x),$
and $\lambda(y,x)$ is a symmetric function of $x,y$.
\item Select $y$ among the $M$ proposals with probability
proportional to $w(y_j,x), j=1,\ldots,M$.
 \item Draw $x_1^*,\ldots,x_{M-1}^*$ variates
from the distribution $T(\cdot|y)$ and let $x_M^*=x$.
\item Accept $y$ with
generalized acceptance probability
$$\rho=\min \left \{1,\frac{w(y_1,x)+\ldots+w(y_M,x)}{w(x_1^*,y)+\ldots+w(x_M^*,y)}
\right \}.$$
\end{enumerate}
\label{alg1}
\end{algo}
\hrule\vspace{5pt}
\end{minipage}
\end{center}

Note that while the MTM uses the same distribution to generate all the proposals, it is possible to extend
this formulation to different proposal distributions without altering the ergodicity of the associated Markov chain.  

Let $T_{j}(x, \cdot)$, with $j=1,\ldots,M$, be a set of proposal distributions
for which  $T_{j}(x,y)>0$ if and only if $T_{j}(y,x)>0$.
Define \[w_{j}(x,y)=\pi(x)T_{j}(x,y)\lambda_{j}(x,y)\qquad j=1,\ldots,M\]
where $\lambda_{j}(x,y)$ is a nonnegative symmetric function in
$x$ and $y$ that can be chosen by the user. The only requirement is that
$\lambda_{j}(x,y)>0$ whenever $T(x,y)>0$. Then the MTM algorithm with different
proposal distributions is given in Algorithm \ref{alg2}.

\begin{center}
\begin{minipage}[thp]{340pt} \par\hrule\vspace{5pt}
\begin{algo}{MTM with Different Proposal Distributions}\label{alg2}
\par\vspace{5pt}\hrule
\begin{enumerate}
\item Draw independently $M$ proposals $y_{1},\ldots,y_{M}$ such that $y_{j}\sim T_{j}(x,\cdot)$.
Compute $w_{j}(y_{j},x)$ for $j=1,\ldots,M$.
\item Select $Y=y$ among the trial set $\{y_{1},\ldots,y_{M}\}$ with probability
proportional to $w_{j}(y_{j},x)$, $j=1,\ldots,M$. Let $J$ be the index
of the selected proposal. Then draw $x_{j}^{*}\sim T_{j}(y,\cdot)$,
$j\neq J$, $j=1,\ldots,M$ and let $x_{J}^{*}=x$.
\item Accept y with probability \[
\rho=\min\left\{ 1,\frac{w_{1}(y_{1},x)+\cdots+w_{M}(y_{M},x)}{w_{1}(x_{1}^{*},y)+\cdots+w_{M}(x_{M}^{*},y)}\right\} \]
 and reject with probability $1-\rho$.
\end{enumerate}
\end{algo}
\hrule\vspace{5pt}
\end{minipage}
\end{center}

It should be noted that Algorithm \ref{alg2} is a special case of the interacting MTM presented in the next section and that the proof of ergodicity for the associated chain follows closely the proof given in Appendix A for the interacting MTM and is not given here.

In Section 4 we will show, through simulation experiments, that this algorithm is more efficient then  a MTM algorithm with a single proposal distribution.



\subsection{General Construction}

Undoubtedly, Algorithm \ref{alg2} offers additional flexibility in organizing the MTM sampler. This section introduces generic methods for using a population of MCMC chains to define the proposal distributions.

\begin{center}
\begin{minipage}[t]{330pt} \par\hrule\vspace{5pt}
\begin{algo}{Interacting Multiple Try Algorithm (IMTM)}
\par\vspace{5pt}\hrule
\begin{itemize}
\item For $i=1,\ldots,N$
\begin{enumerate}
\item Let $x=x_{n}^{(i)}$, for $j=1,\ldots,M_i$  draw $y_{j}\sim T_{j}^{(i)}(\cdot|x_{n}^{(1:i-1)},x,x_n^{(i+1:N)})$ independently and compute
$$w_j^{(i)}(y_j,x)=\pi(y_j)T_j^{(i)}(y_j|x_{n}^{(1:i-1)},x, x_n^{(i+1:N)})\lambda_j^{(i)}(y_j,x).$$
\item Select $J\in\{1,\ldots,M_i\}$ with probability proportional to $w_{j}^{(i)}(y_{j},x)$, $j=1,\ldots,M_i$ and set $y=y_{J}$.
\item For $j=1,\ldots,M_i$ and $j\neq J$ draw $x_{j}^{*}\sim T_{j}^{(i)}(\cdot|x_{n}^{(1:i-1)},y, x_n^{(i+1:N)})$, let $x_{J}^{*}=x_{n}^{(i)}$
and compute $$w_j^{(i)}(x_j^*,y)=\pi(x_j^*)T_j^{(i)}(x_j^*|x_{n}^{(1:i-1)},y, x_n^{(i+1:N)})\lambda_j^{(i)}(x_j^*,y).$$
\item Set $x_{n+1}^{(i)}=y$ with probability \[
\rho_{i}=\min\left\{ 1,\frac{w_{1}^{(i)}(y_{1},x)+\ldots+w_{M_{i}}^{(i)}(y_{M_i},x)}{w_{1}^{(i)}(x_{1}^{*},y)+\ldots+w_{M_i}^{(i)}(x_{M_i}^{*},y)}\right\} \]
 and $x_{n+1}^{(i)}=x_{n}^{(i)}$ with probability $1-\rho_{i}$.
\end{enumerate}
\end{itemize}\label{alg3}
\end{algo}
\hrule\vspace{5pt}
\end{minipage}
\end{center}

Consider a population of  $N$ chains, $X^{(i)}=\{x_n^{(i)}\}_{n\in\mathbb{N}}$ and $i=1,\ldots,N$. For full generality we assume that  the $i$th chain has  MTM transition kernel  with $M_i$ different proposals $\{T_j^{(i)}\}_{1\le j \le M_i}$  (if we set $M_i=1$ we imply that the chain has a MH transition kernel).
The interacting mechanism allows each proposal distribution to possibly depend on the values of the chains  at the previous step. Formally, if $\Xi_{n}=\{x_n^{(i)}\}_{i=1}^{N}$ is the vector of values taken at iteration $n\in\mathbb{N}$ by the population of chains, then we allow each proposal distribution used in updating the population at iteration $n+1$ to depend on $\Xi_n$. The mathematical formalization is used in the description of Algorithm \ref{alg3}. One expects that the chains in the population are spread throughout the sample space and thus the proposals generated
are a good representation of the sample  space $\cay$ ultimately resulting in better mixing for the chain of interest.


In order to give a representation of the IMTM transition density let us introduce the following notation. Let
$T_{j}^{(i)}(x,y)=T_j^{(i)}(y|x_{n}^{(1:i-1)},x, x_n^{(i+1:N)})$,
$T^{(i)}(x,y_{1:M_i})=\prod_{k=1}^{M_i}T_{k}^{(i)}(x,y_{k})$ and
$T_{-j}^{(i)}(x,y_{1:M_i})=\prod_{k\neq j}^{M_i}T_{k}^{(i)}(x,y_{k})$ and define $dy_{1:M_i}=\prod_{k=1}^{M_i}dy_{k}$ and $dy_{-j}=\prod_{k\neq j}^{M_i}dy_{k}$.

The transition density associated to the population of chains is then
\begin{equation}
K(\Xi_{n},\Xi_{n+1})=\prod_{i=1}^{N}K_{i}(x_{n}^{(i)},x_{n+1}^{(i)})
\end{equation}
where
\begin{equation}
K_{i}(x,y)=\sum_{j=1}^{M_i}A_{j}^{(i)}(x,y)T_{j}^{(i)}(x,y)+\left(1-\sum_{j=1}^{M_i}B_{j}^{(i)}(x)\right)\delta_{x}(y)
\end{equation}
is the transition kernel associated to the $i$-th chain of algorithm with
$$
A_{j}^{(i)}(x,y)=\int_{\mathcal{Y}^{2(M_{i}-1)}}\tilde{w}_{j}^{(i)}(y,x)\rho_{j}^{(i)}(x,y)T_{-j}^{(i)}(y,x^{*}_{1:M_{i}})T_{-j}^{(i)}(x,y_{1:M_{i}})dx^{*}_{-j}dy_{-j}
$$
and
$$
B^{(i)}_{j}(x)=\int_{\mathcal{Y}^{2(M_{i}-1)+1}}\rho_{j}^{(i)}(x,y)T_{-j}^{(i)}(y,dx^{*}_{1:M_{i}})T^{(i)}(x,dy_{1:M_{i}})dx^{*}_{-j}dy_{1:M_{i}}.
$$
In the above equations $\tilde{w}_{j}^{(i)}(y_{j},x)=w_{j}^{(i)}(y_{j},x)/(w_{j}^{(i)}(y,x)+\bar{w}_{-k}^{(i)}(y_{1:M_i}|x))$, with $j=1,\ldots,M_{i}$ and
$\bar{w}_{-j}^{(i)}(y_{1:M_{i}}|x)=$ $\sum_{k\neq j}^{M_{i}}w_{k}^{(i)}(y_{k},x)$, are the normalized weights used in the selection step of the IMTM algorithm and $$\rho_{j}^{(i)}(x,y)=\min\left\{1,\frac{w_{j}^{(i)}(y,x)+\bar{w}_{-j}^{(i)}(y_{1:M_i}|x)}{w_{j}^{(i)}(x,y)+\bar{w}_{-j}^{(i)}(x_{1:M_i}^{*}|y)}\right\}$$
is the generalized MH ratio associated to a MTM algorithm.

The validity of Algorithm \ref{alg3} relies upon the detailed balance condition. 
\begin{theo}
The transition density $K_{i}(x_{n}^{(i)},x_{n+1}^{(i)})$ associated to the $i$-th chain of the IMTM algorithm satisfies the conditional detailed balanced condition.
\end{theo}
\begin{proof}
See Appendix A.
\end{proof}

Since the transition $K_{i}(x_{n}^{(i)},x_{n+1}^{(i)})$, $i=1,\ldots,N$ has $\pi(x)$ as stationary distribution and satisfies the conditional detailed balance condition then the joint transition $K(\Xi_{n},\Xi_{n+1})$
has $\pi(x)^{N}$ as a stationary distribution.

An important issue directly connected to the practical implementation of the IMTM is the choice of 
proposal distributions and  the choice of  $\lambda_{j}^{(i)}(x,y)$. First it should be noted that at each iteration of the interacting chains the computational complexity of the algorithm is $\mathcal{O}(N\sum_{i=1}^NM_i)$.
When considering the number of chains and the number of proposals, there are two possible  strategies in designing the interaction mechanism.

The first strategy is to use a small number of chains, say  $2\le N\le 5$, 
in order to improve the mixing of each chain and to allow for large jumps between different
regions of the state space. When applying this strategy to our IMTM
algorithms it is possible to set the number of proposals to be equal to the number of chains, i.e. $M_i=N$, for all $1\le i\le N$. In this way all
the chains can interact at each iteration of the algorithm and many search directions can be included among the proposals.

A second  strategy is to use a higher number of chains, e.g. $N=100$, in order to possibly have, at each iteration, a good
approximation of the target or  a much higher number of search directions for a good exploration of the sample space. This algorithm design
strategy is common in Population Monte Carlo or Interacting MCMC methods. Clearly when a high number of chains is used within IMTM, it is necessary to set $M_i<N$.
In the next section we discuss a few strategies to built the $M_i$ proposals.

\subsection{Parsing the Population of Auxiliary Chains}

One of the strategies that revealed to be successful in our applications consists in the random selection of a certain number of chains of the population
in order to build the proposals. More specifically, we let $M_{i}=M$, for all $i$, and when updating the $i$-th chain of the population we sample $I_{1},\ldots,I_{M}$ random
indexes from the uniform distribution $\mathcal{U}_{\{1,\ldots,N\}}$, with $N>M$, and then set the proposals: $T_{j}^{(i)}(y,x)=T_{j}^{(i)}(\cdot|x,x_{n}^{(I_{1})},\ldots, x_{n}^{(I_{M})})$, for all $j=1,\ldots,M$. On the basis of our simulation experiments we found that the following choice $T_{j}^{(i)}(\cdots|x_{n}^{(I_{1})},\ldots, x_{n}^{(I_{M})})=T_{j}(\cdot|x_{n}^{(I_{j})})$ is works well in improving the mixing of the chains.

Previously suggested forms for the function $\lambda_{j}^{(i)}(x,y)$ \citep{Liu:2000cr} are:
\begin{enumerate}
\item [a)] $\lambda_{j}^{(i)}(x,y)=1$
\item [b)] $\lambda_{j}^{(i)}(x,y)=2\{T_{j}^{(i)}(x,y)+T_{j}^{(i)}(y,x)\}^{-1}$
\item [c)] $\lambda_{j}^{(i)}(x,y)=\{T_{j}^{(i)}(x,y)T_{j}^{(i)}(y,x)\}^{-\alpha}$, $\alpha>0$.
\end{enumerate}
Here we propose to include in the choice of $\lambda$ the information provided by the population of chains. Therefore  we suggest to modify the above functions as follows
\begin{enumerate}
\item [a$^\prime$)] $\lambda_{j}^{(i)}(x,y)=\nu_{j}$
\item [b$^\prime$)] $\lambda_{j}^{(i)}(x,y)=2\nu_{j}\,\{T_{j}^{(i)}(x,y)+T_{j}^{(i)}(y,x)\}^{-1}$
\item [c$^\prime$)] $\lambda_{j}^{(i)}(x,y)=\nu_{j}\,\{T_{j}^{(i)}(x,y)T_{j}^{(i)}(y,x)\}^{-\alpha}$, $\alpha>0$
\end{enumerate}
where the factor $\nu_{j}$ captures the behaviour of the auxiliary chains at the previous iteration
$$
\nu_{j}=\frac{1}{N}\sum_{i=1}^{N}\mathbb{I}_{\{j\}}(J_{n-1}^{(i)}),\quad j=1,\ldots,M
$$
where $J_{n-1}^{(i)}$ is the random index of the selection step at the iteration $n-1$ for the $i$-th chain.  The modifications proposed for $\lambda(\cdot,\cdot)$ would increase the use of those proposal distributions favoured by the population of chains at previous iteration. Since  $\nu_j$ depends only on samples generated at the previous step by the population of chains, the ergodicity of the  IMTM chain is preserved. An alternative strategy is to sample the random indexes $I_{1},\ldots,I_{M}$ with probabilities proportional to $\nu_{j}$.

\subsection{Annealed IMTM}

Our belief in IMTM's improved performance is underpinned by the assumption that the population of Monte Carlo chains is spread throughout the sample space. This can be partly achieved by initializing the chains using draws from a distribution overdispersed with respect to $\pi$ \citep[see also][]{jenn,gelmrubin}  and partly by modifying the stationary distribution for some of the chains in the population.
Specifically, we consider the sequence of annealed distributions $\pi_t=\pi^{t}$ with $t \in \{\xi_1 , \xi_2,\ldots, \xi_N\}$, where $1= \xi_1 > \xi_2> \ldots>\xi_n$, for instance $\xi_t=1/t$. When $t,s$ are close temperatures, $\pi_t$ is similar to $\pi_s$, but $\pi=\pi_{1}$ may be much harder to sample from than $\pi_{\xi_N}$ as has been long recognized in the simulated annealing and simulated tempering literature \citep[see][]{mapa,gey-temp,neal-temp}. Therefore, it is likely that some of the chains designed to sample from $\pi_1,\ldots,\pi_N$ have good mixing properties, making them good candidates for the population of MCMC samplers needed for the IMTM.

We  thus consider  the Monte Carlo population made of the $N-1$ chains having $\{\pi_2,\ldots, \pi_N\}$ as stationary distributions. An example of annealed interacting MTM is given in Algorithm \ref{alg5bis}. Note that we let the $i$-th chain to interact only with the $N-i+1$ chains at  higher temperature by sampling $I_{1},\ldots,I_{M}$ from $\mathcal{U}_{\{1,\ldots,N-i+1\}}$. 

An astute reader may have noticed that the use of MTM for {\it each} auxiliary chain may be redundant since for smaller $\xi_i$'s the distribution $\pi_i$ is easy to sample from. In Algorithm \ref{alg5} we present an alternative 
implementation of the annealed IMTM in which each auxiliary chain is  MH with target $\pi_i$, $1\le i \le N-1$.

\begin{center}
\begin{minipage}[thp]{330pt} \par\hrule\vspace{5pt}
\begin{algo}{Annealed IMTM Algorithm (AIMTM1)}
\par\vspace{5pt}\hrule
\begin{itemize}
\item For $i=1,\ldots,N$
\begin{enumerate}
\item Let $x=x_{n}^{(i)}$ and sample $I_{1},\ldots,I_{M}$ from $\mathcal{U}_{\{1,\ldots,N-i+1\}}$.
\item For $j=1,\ldots,M$  draw $y_{j}\sim T_{j}^{(i)}(\cdot|x_{n}^{(I_{j})})$ independently and compute
$$w_j^{(i)}(y_{j},x)=\pi(y_j)T_{j}^{(i)}(y_j|x_{n}^{(I_{j})})\lambda_{j}^{(i)}(y_{j},x).$$
\item Select $J\in\{1,\ldots,M\}$ with probability proportional to $w_{j}^{(i)}(y_{j},x)$, $j=1,\ldots,M$ and set $y=y_{J}$.
\item For $j=1,\ldots,M$ and $j\neq J$ draw $x_{j}^{*}\sim T_{j}^{(i)}(\cdot|x_{n}^{(1:i-1)},y, x_n^{(i+1:N)})$, let $x_{J}^{*}=x_{n}^{(i)}$
and compute $$w_j^{(i)}(x_{j}^{*},y)=\pi(x_{j}^{*})T_j^{(i)}(x_{j}^{*}|y)\lambda_j^{(i)}(x_{j}^{*},y).$$
\item Set $x_{n+1}^{(i)}=y$ with probability $\rho_{i}$, where $\rho_{i}$ is the generalized M.H. ratio of the IMT algorithm
 and $x_{n+1}^{(i)}=x_{n}^{(i)}$ with probability $1-\rho_{i}$.
\end{enumerate}
\end{itemize}
\label{alg5bis}
\end{algo}
\hrule\vspace{5pt}
\end{minipage}
\end{center}

\begin{center}
\begin{minipage}[thp]{330pt} \par\hrule\vspace{5pt}
\begin{algo}{Annealed IMTM Algorithm (AIMTM2)}
\par\vspace{5pt}\hrule
\begin{itemize}
\item For $i=1$
\begin{enumerate}
\item Let $x=x_{n}^{(i)}$ and sample $I_{1},\ldots,I_{M}$ from $\mathcal{U}_{\{1,\ldots,N\}}$.
\item For $j=1,\ldots,M$  draw $y_{j}\sim T_{j}^{(i)}(\cdot|x_{n}^{(I_{j})})$ independently and compute
$$w_j^{(i)}(y_{j},x)=\pi(y_j)T_{j}^{(i)}(y_j|x_{n}^{(I_{j})})\lambda_{j}^{(i)}(y_{j},x).$$
\item Select $J\in\{1,\ldots,M\}$ with probability proportional to $w_{j}^{(i)}(y_{j},x)$, $j=1,\ldots,M$ and set $y=y_{J}$.
\item For $j=1,\ldots,M$ and $j\neq J$ draw $x_{j}^{*}\sim T_{j}^{(i)}(\cdot|x_{n}^{(1:i-1)},y, x_n^{(i+1:N)})$, let $x_{J}^{*}=x_{n}^{(i)}$
and compute $$w_j^{(i)}(x_{j}^{*},y)=\pi(x_{j}^{*})T_j^{(i)}(x_{j}^{*}|y)\lambda_j^{(i)}(x_{j}^{*},y).$$
\item Set $x_{n+1}^{(i)}=y$ with probability $\rho_{i}$, where $\rho_{i}$ is the generalized M.H. ratio of the IMT algorithm
 and $x_{n+1}^{(i)}=x_{n}^{(i)}$ with probability $1-\rho_{i}$.
\end{enumerate}
\item For $i=2,\ldots,N$
\begin{enumerate}
\item Let $x=x_{n}^{(i)}$ and update the proposal function $T^{(i)}(\cdot|x)$.
\item Draw $y\sim T^{(i)}(\cdot|x)$ and compute
$$\rho_{i}=\min\left\{1,\frac{\pi(y)^{\xi_{i}}T^{(i)}(x|y)}{\pi(x)^{\xi_{i}}T^{(i)}(y|x)}\right\}.$$
\item Set $x_{n+1}^{(i)}=y$ with probability $\rho_{i}$ and $x_{n+1}^{(i)}=x_{n}^{(i)}$ with probability $1-\rho_{i}$.
\end{enumerate}
\end{itemize}\label{alg5}
\end{algo}
\hrule\vspace{5pt}
\end{minipage}
\end{center}

The chain of interest, corresponding to $\xi=1$, has an MTM transition kernel with $M=N-1$ proposal distributions. At time $n$ the $j$th proposal distribution used for the chain ergodic to $\pi$  is 
$T_{j}^{(1)}=T^{(j)}$, is the same as the proposal used by the $j$th auxiliary chain, for all $2\le j \le N$.
%



An additional gain could be obtained if the auxiliary chains' transition kernels are
modified using adaptive MCMC strategies \citep[see also][for another example of adaption for interacting chains]{Chauveau2002}. However, it should be noted that letting the auxiliary chains adapt indefinitely results in complex theoretical justifications for the IMTM which go beyond the scope of this paper and will be presented elsewhere.
Our current recommendation is to use finite adaptation for  the auxiliary chains prior to the start of the IMTM. One could take advantage of multi-processor computing units and use parallel programming  to increase the computational efficiency of this approach.

The adaptation of  $\lambda^{(i)}_{j}(\cdot,\cdot)$, through the weights $\nu_j$ defined in the previous section, should be used cautiously in this case. The aim of the annealing procedure is to allow the higher temperatures chains to explore widely the sample space and to improve the mixing of the MTM chain. Using $\nu_j$ the context of annealed IMTM  could arbitrarily penalize some of the higher temperature proposals and reduce the effectiveness of the annealing strategy.

It is possible to have a Monte Carlo approximation of a quantity of interest by using the output produced by all the chains in the population.
For example let $$\mathcal{I}=\int_{\mathcal{Y}} h(x) \pi(x)dx$$ be the quantity of interest where $h$ is a test function. It is possible to approximate this quantity as follows
$$\mathcal{I}_{NT}=\frac{1}{T}\sum_{n=1}^{T}\frac{1}{\bar{\zeta}}\sum_{j=1}^{N}h(x_{n}^{(j)})\zeta_{j}(x_{n}^{(j)})$$
where $x_{n}^{(i)}$, with $n=1,\ldots,T$ and $i=1,\ldots,N$ is the output of a IMTM chains with target $\pi^{\xi_{i}}$ and $\zeta_{j}(x)=\pi(x)/\pi^{\xi_{j}}(x)$ is a set of importance weights with normalizing constant $\bar{\zeta}=\sum_{j=1}^{N}\zeta_{j}(x_{n}^{(j)})$.
\subsection{Gibbs within IMTM update}

It should be noticed that in the proposed algorithm at the $n$-th iteration the $N$ chains are updated
simultaneously. In the interacting MCMC literature a sequential updating scheme (Gibbs-like updating)
has been proposed for example in \cite{Mengersen:Robert:03} and \cite{Campillo:Rakotozafy:Rossi:09}.
In the following we show that the Gibbs-like updating also apply to our IMTM context. In the
Gibbs-like interacting MTM (GIMTM) algorithm given in Algorithm \ref{alg4} the different proposals functions of
the $i$-th chain, with $i=1,\ldots,N$, may depend on the current values of the updated chains $x_{n+1}^{(j)}$, with $j=1,\ldots,(i-1)$
and on the last values $x_{n}^{(j)}$, with $j=(i+1),\ldots,N$, of the chains which have not yet been updated.

\begin{center}
\begin{minipage}[thp]{340pt} \par\hrule\vspace{5pt}
\begin{algo}{Gibbs-like IMTM Algorithm (GIMTM)}\label{alg4}
\par\vspace{5pt}\hrule
\begin{itemize}
\item For $i=1,\ldots,N$
\begin{enumerate}
\item For $j=1,\ldots,M_{i}$  draw $y_{j}\sim T_{j}^{(i)}(\cdot|x_{n+1}^{(1:i-1)},x_n^{(i)},x_n^{(i+1:N)})$ independently and compute
$$w_j^{(i)}(y_j,x)=\pi(y_j)T_j^{(i)}(y_j|x_{n+1}^{(1:i-1)},x_n^{(i)}, x_n^{(i+1:N)})\lambda_j^{(i)}(y_j,x_n^{(i)}).$$
\item Select $J\in\{1,\ldots,M_{i}\}$ with probability proportional to $w_{j}^{(i)}(y_{j},x)$, $j=1,\ldots,M_{i}$ and set $y=y_{J}$.
\item For $j=1,\ldots,M_{i}$ and $j\neq J$ draw
$x_{j}^{*}\sim T_{j}^{(i)}(\cdot|x_{n+1}^{(1:i-1)},y, x_n^{(i+1:N)})$, let $x_{J}^{*}=x_{n}^{(i)}$ and compute
$$w_j^{(i)}(x_j^*,y)=\pi(x_j^*)T_j^{(i)}(x_j^*|x_{n+1}^{(1:i-1)},y, x_n^{(i+1:N)})\lambda_j^{(i)}(x_j^*,y).$$
\item Set $x_{n+1}^{(i)}=y$ with probability \[
\rho_{i}=\min\left\{ 1,\frac{w_{1}^{(i)}(y_{1},x)+\ldots+w_{M_{i}}^{(i)}(y_{M_{i}},x)}{w_{1}^{(i)}(x_{1}^{*},y)+\ldots+w_{M_{i}}^{(i)}(x_{M_{i}}^{*},y)}\right\} \]
 and $x_{n+1}^{(i)}=x_{n}^{(i)}$ with probability $1-\rho_{i}$.
\end{enumerate}
\end{itemize}
\end{algo}
\hrule\vspace{5pt}
\end{minipage}
\end{center}

In the GIMTM algorithm the iteration mechanism between the chains is not the same as in the IMTM algorithm. The chains are no longer independent since the proposals  may depend on the current values for some of  the  chains in the population. The transition kernel for the whole population is
$$K(\Xi_{n},\Xi_{n+1})=\prod_{i=1}^{N}K_{i}(x_{n}^{(i)},x_{n+1}^{(i)}|x_{n+1}^{(1:i-1)},x_{n}^{(i+1:N)})$$
and in this case the validity of the algorithm still relies upon the conditional detail balance condition given for the IMTM algorithm.

Finally we remark that the GIMTM algorithm allows us to introduce further possible choices for the $\lambda_{j}^{(i)}(x,y)$ functions.
In particular a repulsive factor (see \cite{Mengersen:Robert:03}) can be introduced in the selection weights in order to induce
negative dependence between the chains. We let the study of the GIMTM algorithm and the use of repulsive factors for future research and
focus instead on the properties of the IMTM algorithm.
%

\section{Some generalizations}
In the following we will discuss some possible generalization of the IMTM algorithm. First we show how to use the stochastic overrelaxation method to possibly have a further gain in the efficiency. Secondly we suggest two possible strategies to built the different proposal functions of the IMTM. The first strategy consists in proposing values along different search directions and represents an extension of the random-ray Monte Carlo algorithm presented in \cite{Liu:2000cr}. The second strategy relies upon a suitable combination of target tempering and adaptive MCMC chains.

\subsection{Stochastic Overrelaxation}

Stochastic overrelaxation (SOR) is a Markov chain Monte Carlo technique
developed by Adler (1981) for normal densities and subsequently extended
by Green and Han (1992) for non-normal targets. The idea behind this
approach is to induce negative correlation between consecutive draws
of a single MCMC process.

Within the MTM algorithm we can implement SOR by inducing negative correlation between the
proposals and between the proposals and the current state of the chain, $x$.
A natural and easy to implement procedure may be based on the assumption that $(y_1,\ldots,y_{M-1},x)^T \sim N_{d\times M}(\mathbf{0},
V)$ where $V$'s structure is dictated  by the desired negative dependence between the
proposals $y_1,\ldots,y_n$'s and  $x$, specifically

\[
V=\left (
\begin{array}{cccc}
\Sigma_1 & \Psi_{12} & \ldots &\Psi_{1M}\\
\Psi_{12} & \Sigma_2 & \ldots & \Psi_{2M}\\
\ldots & \ldots & \ldots & \ldots \\
\Psi_{1M}& \Psi_{2M} & \ldots & \Sigma_M\\
\end{array}
\right ) .
\]

For instance, we can set $\Psi_{ij}=0$ whenever $i,j\ne M$ and
$\Psi_{iM}=\Sigma_i^{1/2}R_{iM} \Sigma_M^{1/2}$ where
$R_{iM}$ is a correlation matrix which corresponds to extreme negative
correlation \citep[see][for a discussion of extreme dependence]{crameng2}
between any two components (with same index) of
$y_i$ and $x$, for any $1\le i\le M-1$. The general construction falls within the context of dependent proposals as discussed
by \cite{Craiu:2007rr} with the additional bonus of
"pulling"  the proposals away  from the current state  due to the imposed negative correlation.
This essentially ensures that no proposals are exceedingly close to the
current location of the chain.
Also note that the construction is general enough and  can be applied
for Algorithms 2 and 3 as long as the proposal distributions are Gaussian.

\subsection{Multiple Random-ray Monte Carlo}
We show here that the use of different proposals for the MTM algorithm allows also to extend the random-ray Monte Carlo method given in \cite{Liu:2000cr}.
In particular the proposed algorithm allows to deal with multiple search directions at each iteration of the chains. At the $n$-th iteration of the chain, in order
to update the set of chains $\Xi_{n}$, the algorithm performs for each chain $x_{n}^{(r)}\in\Xi_{n}$,
with $r=1,\ldots,N$, the following steps:
\begin{enumerate}
\item Evaluate the gradient $\log\pi(x)$ at $x_{n}^{(r)}$ and find the mode
$a_{n}$ along $x_{n}^{(r)}+ru_{n}$ where $u_{n}=x_{n}^{(r)}-x_{n-1}^{(r)}$.

\item Sample $I_{1},\ldots,I_{M}$ from the uniform $\mathcal{U}_{\{1,\ldots,r-1,r+1,\ldots,N\}}$.

\item Let $e_{n,j}=(a_{n}-x_{n}^{(I_{j})})/||a_{n}-x_{n}^{(I_{j})}||$ and sample $r_{j}$ from
$\mathcal{N}(0,\sigma^{2})$.
\end{enumerate}
and then use the set of proposals $T_{j}$ which depends on $e_{n,j}$ to perform a MTM transition with different proposals as in the IMTM algorithm (see Alg. \ref{alg1}).

\section{Simulation Results}
\subsection{Single-chain results}
In this section we carry out, through some examples, a comparison between the single-chain multiple try algorithms MTM-DP with different proposals and the
algorithm in Liu et al. (2000). In the MTM-DP algorithm we consider four Gaussian random-walk proposals $\mathbf{y}_{j}\sim\mathcal{N}_{n}(\mathbf{x},\Lambda_{j})$ with
$\Lambda_{1}=0.1Id_{n}$, $\Lambda_{2}=5Id_{n}$, $\Lambda_{3}=50Id_{n}$ and $\Lambda_{4}=100Id_{n}$, where $Id_{n}$ denotes the $n$-order identity matrix. In the
MTM selection weights we set $\lambda_{j}(\mathbf{x},\mathbf{y})=2\alpha_{j}/(T_{j}(\mathbf{x},\mathbf{y})+T_{j}(\mathbf{y},\mathbf{x}))$,
where $\alpha_{j}=0.25$, for $j=1,\ldots,4$.

\bigskip

In order to compare the MTM algorithm with different proposals and the Multiple Try algorithm of Liu et al. (2000) we consider 20,000
iterations and use four trials $\mathbf{y}_{j}$ generated by the following proposal
\begin{equation}\label{Eq_ProposalIMT}
T_{j}(\mathbf{x},\mathbf{y}_{j})=\sum_{j=1}^{4}\alpha_{j}\mathcal{N}_{n}(\mathbf{x},\Lambda_{j})
\end{equation}
where $\Lambda_{j}$, $j=1,\ldots,4$, have been defined above. In the weighs of the selection step we
set $\lambda(\mathbf{x},\mathbf{y})=2/(T_{j}(\mathbf{x},\mathbf{y})+T_{j}(\mathbf{y},\mathbf{x}))$.

\subsubsection{Bivariate Mixture with two components}
We consider the following bivariate mixture of two normals
\begin{equation}
\frac{1}{3}\mathcal{N}_{2}(\boldsymbol{\mu}_{1},\Sigma_{1})+\frac{2}{3}\mathcal{N}_{2}(\boldsymbol{\mu}_{2},\Sigma_{2})
\end{equation}
with $\boldsymbol{\mu}_{1}=(0,0)'$, $\boldsymbol{\mu}_{2}=(10,10)'$, $\Sigma_{1}=\hbox{diag}(0.1,0.5)$ and $\Sigma_{2}=\hbox{diag}(0.5,0.1)$.

In Fig. \ref{Fig_Acf} the ACF with 30 lags for the first component of the bivariate MH chain. The autocorrelation is lower for the MTM algorithm with different proposals. 

\begin{figure}[p]
\centering
\includegraphics[width=270pt]{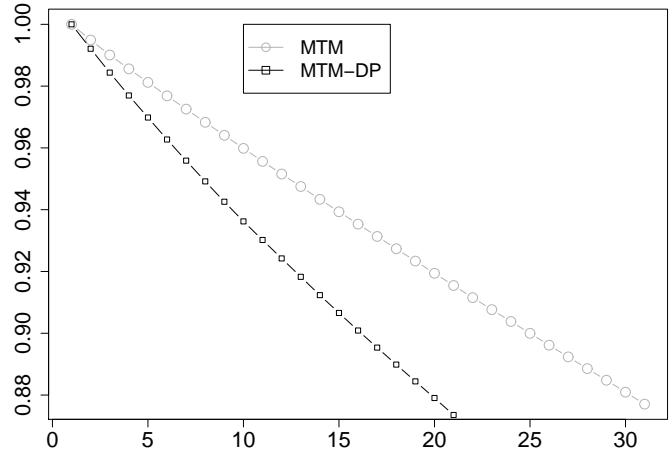}
\caption{Autocorrelation function of the Liu's MTM (gray line) and MTM with different proposals (black line).} \label{Fig_Acf}
\end{figure}
\begin{figure}[p]
\centering
\includegraphics[width=270pt]{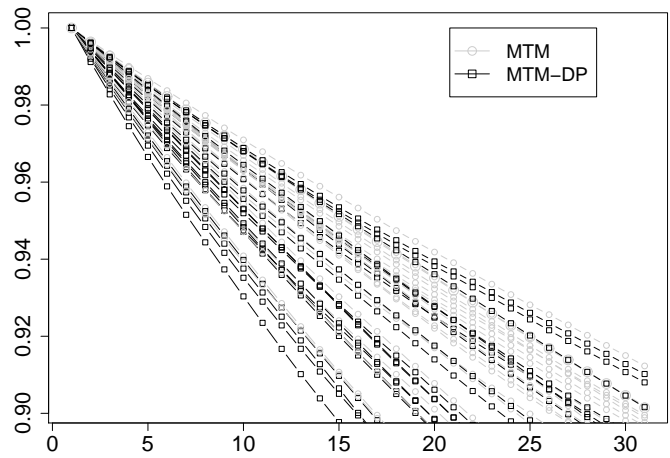}
\caption{Autocorrelation function (ACF) of the 20 components of the multivariate MH chain for the Liu's MTM (gray lines) and MTM with different proposals (black lines).} \label{Fig_AcfMultiv}
\end{figure}

\subsubsection{Multivariate Normal Mixture}
We compare the algorithms for high-dimensional targets. We consider the following multivariate mixture of two normals with a sparse variance-covariance
structure
\begin{equation}
\frac{1}{3}\mathcal{N}_{20}(\boldsymbol{\mu}_{1},\Sigma_{1})+\frac{2}{3}\mathcal{N}_{20}(\boldsymbol{\mu}_{2},\Sigma_{2})
\end{equation}
with $\boldsymbol{\mu}_{1}=(3,\ldots,3)'$, $\boldsymbol{\mu}_{2}=(10,\ldots,10)'$ and $\Sigma_{j}$, with $j=1,2$, generated independently from
a Wishart distribution $\Sigma_{j}\sim\mathcal{W}_{20}(\nu,Id_{20})$ where $\nu$ is the degrees of freedom parameter of the Wishart. In the experiments
we set $\nu=21$.

The autocorrelation function of the chain for one of the experiment is given in Fig. \ref{Fig_AcfMultiv}. The ACF has been evaluated for each components of the $20$-dimensional chain. The values of the ACF of the MTM-DP (black lines in Fig. \ref{Fig_AcfMultiv}) are less than those of the original MTM (gray lines of the same figure) in all the directions of the support space. We conclude that the MTM algorithm with different proposals (Algorithm \ref{alg3})  outperforms the \cite{Liu:2000cr} MTM algorithm.

\subsection{Multiple-chains results}
In this section we show real and simulated data results of the general interacting multiple try algorithm in Alg. \ref{alg1}.
\subsubsection{Bivariate Mixture with two components}
The target distribution is the following bivariate mixture of two normals
\begin{equation}
\frac{1}{3}\mathcal{N}_{2}(\boldsymbol{\mu}_{1},\Sigma_{1})+\frac{2}{3}\mathcal{N}_{2}(\boldsymbol{\mu}_{2},\Sigma_{2})
\end{equation}
with $\boldsymbol{\mu}_{1}=(0,0)'$, $\boldsymbol{\mu}_{2}=(10,10)'$, $\Sigma_{1}=\hbox{diag}(0.1,0.5)$ and $\Sigma_{2}=\hbox{diag}(0.5,0.1)$.

We consider a population of $N=50$ chains with $M=50$ proposals and 1,000 iterations of the IMTM algorithm. For each chain we consider the case $T_{j}^{(i)}(\mathbf{y}|\mathbf{x}_{n+1}^{(1:i-1)},\mathbf{x}_{n}^{(i)},\mathbf{x}_{n}^{(i+1:N)})=T^{(j)}(\mathbf{y},\mathbf{x}_{n}^{(j)})$ and draw
\begin{equation}
\mathbf{y}_{j}\sim\mathcal{N}_{2}(\mathbf{x}_{n}^{(j)},\Lambda_{j})
\end{equation}
where $\Lambda_{j}=(0.1+5j)Id_{2}$.
In this specification of the IMTM algorithm each chain has $M$ independent proposals with conditional mean given by the previous values of the chains in the population.

In this experiment we consider two kind of weights. First we set $\lambda_{j}^{(i)}(x,y)=(T_{j}^{(i)}(x,y)T_{j}^{(i)}(y,x))^{-1}$, that corresponds to use importance sampling selection weights, secondly we consider $\lambda_{j}^{(i)}(x,y)=2(T_{j}^{(i)}(x,y)+T_{j}^{(i)}(y,x))^{-1}$, which implies a symmetric MTM algorithm. We denote with IMTM-IS and IMTM-TA respectively the resulting algorithms.

Fig. \ref{Fig_SurfIMT} show the results of the IMTM-IS and IMTM-TA algorithms at the last iteration of the population of chains (black dots). In both of the algorithms the population is visiting the two modes of the distribution in the right proportion. Moreover each chain is able to jump from one mode to the other. (the light-gray line represents the sample path of one of the chain).

\begin{figure}[p]
\centering
\includegraphics[width=320pt]{SurfaceLeiInterJ50N1000IS}
\includegraphics[width=320pt]{SurfaceLeiInterJ50N1000TA}
\caption{Log-target level sets (solid lines), the MH chains (red dots) at the last iteration of the IMTM algorithms and the path of one of the chain of the set (gray line). Up: output of the IMTM-IS with importance sampling selection weights. Bottom: ouput of the IMTM-TA with symmetric lambda as in (c).} \label{Fig_SurfIMT}
\end{figure}

\subsubsection{Beta-Binomial Model}
We consider here the problem of the genetic instability of esophageal cancers. During a neoplastic progression
the cancer cells undergo a number of genetic changes and possibly lose entire chromosome sections. The loss of a
chromosome section containing one allele by abnormal cells is called \textit{Loss of Heterozygosity} (LOH). The LOH can
be detected using laboratory assays on patients with two different alleles for a particular gene. Chromosome regions
containing genes which regulate cell behavior, are hypothesized to have a high rates of LOH. Consequently the loss
of these chromosome sections disables important cellular controls.

Chromosome regions with high rates of LOH are hypothesized to contain \textit{Tumor Suppressor Genes} (TSGs), whose
deactivation contributes to the development of esophageal cancer. Moreover the neoplastic progression is thought to
produce a high level of background LOH in all chromosome regions.

In order to discriminate between "background" and TSGs LOH, the Seattle Barrett's Esophagus research project
(\cite{barret06}) has collected LOH rates from esophageal cancers for 40 regions, each on a distinct chromosome
arm. The labeling of the two groups is unknown so \cite{desai} suggest to consider a mixture model for the frequency
of LOH in both the "background" and TSG groups.

We consider the hierarchical Beta-Binomial mixture model proposed in \cite{Warnes:01}
\begin{eqnarray}
&&f(x,n|\eta,\pi_{1},\pi_{2},\gamma)=\eta{n\choose x}\pi_{1}^{x}(1-\pi_{1})^{n-x}+\\
&&\quad(1-\eta){n\choose x}\frac{\Gamma(1/\omega_{2})}{\Gamma(\pi_{2}/\omega_{2})\Gamma((1-\pi_{2})/\omega_{2})}
\frac{\Gamma(x+\pi_{2}/\omega_{2})\Gamma(n-x+(1-\pi_{2})/\omega_{2})}{\Gamma(n+1/\omega_{2})}\nonumber
\end{eqnarray}
with $x$ number of LOH sections, $n$ the number of examined sections, $\omega_{2}=\exp\{\gamma\}/(2(1+\exp\{\gamma\}))$.
Let $\mathbf{x}=(x_{1},\ldots,x_{m})$ and $\mathbf{n}=(n_{1},\ldots,n_{m})$ be a set of observations from $f(x,n|\eta,\pi_{1},\pi_{2},\gamma)$
and let us assume the following priors
\begin{eqnarray}
\eta\sim\mathcal{U}_{[0,1]},\quad \pi_{1}\sim\mathcal{U}_{[0,1]},\quad \pi_{2}\sim\mathcal{U}_{[0,1]}\quad\hbox{and}\quad\gamma\sim\mathcal{U}_{[-30,30]}
\end{eqnarray}
with $\mathcal{U}$ the uniform distribution on $[a,b]$. Then the posterior distribution is
\begin{equation}
\pi(\eta,\pi_{1},\pi_{2},\gamma|\mathbf{x},\mathbf{n})\propto\prod_{j=1}^{m}f(x_{j},n_{j}|\eta,\pi_{1},\pi_{2},\gamma)
\end{equation}
The parametric space is of dimension four: $(\eta,\pi_{1},\pi_{2},\gamma)\in[0,1]^{3}\times[-30,30]$ and the posterior distribution has two well-separated  modes making it difficult to sample using generic methods.

\begin{figure}[t]
\centering
\includegraphics[width=300pt]{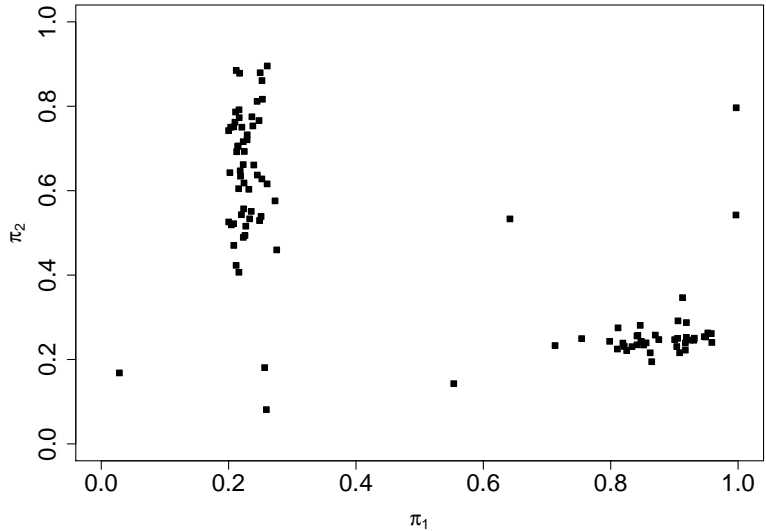}\label{Fig_Barret}
\caption{Values of the population of chains (dots) at the last iteration on the subspace ($\pi_{1}$,$\pi_{2}$). The interaction is given
by $M=4$ proposal functions randomly selected between the population of $N=100$ chains.}
\end{figure}

We apply the IMTM algorithm with $100$ iterations, $M=4$ proposal functions randomly selected between a population of $N=100$ chains. For each chain we consider importance sampling weights in the selection step, i.e. we set $\lambda_{j}^{(i)}(x,y)=(T_{j}^{(i)}(x,y)T_{j}^{(i)}(y,x))^{-1}$ with $j=1,\ldots,4$ and $i=1,\ldots,100$. The values of the population of chains (dots) at the last iteration on the subspace ($\pi_{1}$,$\pi_{2}$) is given in Fig. \ref{Fig_Barret}. The IMTM is able to visit both  regions of the parameter space confirming the analysis of \cite{craiu-jeff-yang} and  \cite{Warnes:01}.

\subsubsection{Stochastic Volatility}
The estimation of the stochastic volatility (SV) model due to \cite{Taylor:94} still represents a challenging issue in both off-line (\cite{Celeux:Marin:Robert:06}) and sequential (\cite{Casarin:Marin:09}) inference contexts. One of the main difficulties is due to the  high dimension of the sampling space which hinders the use of the data-augmentation and prevents a reliable joint estimation of the parameters and the latent variables. As highlighted in \cite{Casarin:Marin:Robert:09} using multiple chains with a chain interaction mechanism could lead to a substantial improvement in the MCMC method for this kind of model. We consider the SV model given in \cite{Celeux:Marin:Robert:06}
\begin{eqnarray*}
y_t|h_t                          & \sim & \mathcal{N}\left(0,e^{h_t}\right) \\
h_t|h_{t-1},\boldsymbol{\theta}  & \sim & \mathcal{N}\left(\alpha+\phi h_{t-1},\sigma^{2}\right) \\
h_0|\boldsymbol{\theta}          & \sim & \mathcal{N}\left(0,\sigma^{2}/(1-\phi^{2})\right)
\end{eqnarray*}
with $t=1,\ldots,T$ and $\boldsymbol{\theta}=$ $(\alpha$, $\phi$, $\sigma^{2})$. For the parameters we assume the noninformative prior \citep[see][]{Celeux:Marin:Robert:06} $$\pi(\boldsymbol{\theta})\propto 1/(\sigma\beta)\mathbb{I}_{(-1,1)}(\phi)$$ where $\beta^{2}=\exp(\alpha)$.
In order to simulate from the posterior we consider the full conditional distributions and apply a Gibbs algorithm. If we  define  $\mathbf{y}=(y_{1},\ldots,y_{T})$ and $\mathbf{h}=(h_{0},\ldots,h_{T})$ then the full conditionals for $\beta$ and $\phi$ are the inverse gamma distributions
\begin{eqnarray*}
\beta^2|\mathbf{h},\mathbf{y}\quad & \sim    & \mathcal{I}\mathcal{G}\left(\sum_{t=1}^{T}y^{2}_t\exp(-h_t)/2,(T-1)/2\right)\, \\
\sigma^2|\phi,\mathbf{h},\mathbf{y}  & \sim    & \mathcal{I}\mathcal{G}\left(\sum_{t=2}^{T}(h_t-\phi h_{t-1})^{2}/2+h^{2}_{1}(1-\phi^{2}),(T-1)/2\right)
\end{eqnarray*}
 and $\phi$ and the latent variables have non-standard full conditionals
\begin{eqnarray*}
&&\!\!\!\!\!\!\!\!\!\pi(\phi|\sigma^{2},\mathbf{h},\mathbf{y})\propto (1-\phi^{2})^{1/2}\exp\left(-\frac{\phi^{2}}{2\sigma^{2}}\sum_{t=2}^{T-1}h^{2}_t-\frac{\phi}{\sigma^{2}}\sum_{t=2}^{T}h_t h_{t-1}\right)\mathbb{I}_{(-1,+1)}(\phi)\, \\
&&\!\!\!\!\!\!\!\!\!\pi(h_t|\alpha,\phi,\sigma^{2},\mathbf{h},\mathbf{y})  \propto \exp\left(-\frac{1}{2\sigma^{2}}\left((h_t-\alpha-\phi h_{t-1})^{2}-\right.\right. \\
&&\quad\quad\quad\quad\quad\quad\quad\quad\quad\quad\left.\left.(h_{t+1}-\alpha-\phi h_t)^{2}\right)-\frac{1}{2}\left(h_t+y^{2}_t\exp(-h_t)\right)\right)  .
\end{eqnarray*}
 In order to sample from the posterior we
use an IMTM within Gibbs algorithm. A detailed description of the
proposal distributions for $\phi$ and $h_t$ can be found in
\cite{Celeux:Marin:Robert:06}.

We consider the two parameter settings $(\alpha$,\,$\phi$,$\,\sigma^{2})=$ $(0,0.99,0.01)$ and $(\alpha$,$\,\phi$,$\,\sigma^{2})=$ $(0,0.9,0.1)$ which correspond, in a financial stock market context, to daily and weekly frequency data respectively. Note that as reported in \cite{Casarin:Marin:09} inference in the daily example is more difficult. We compare the performance of  MH within Gibbs and IMTM within Gibbs algorithms in terms of Mean Square Error (MSE) for the parameters and  of cumulative RMSE for the latent variables. We carry out the comparison in statistical terms and estimate the MSE and RMSE by running the algorithms on 20 independent simulated datasets of 200 observations. In the comparison we take into account the computational cost and for the IMTM within Gibbs we use $N=20$ interacting chains, 1,000 iterations and $M=5$ proposal functions. This setting corresponds to 100,000 random draws and is equivalent to the 100,000 iterations of the MH within Gibbs algorithm used in the comparison. Note that the proposal step of the IMTM selects at random the proposal functions between the other chains and the selection step uses $\lambda_{j}^{(i)}(x,y)=(T_{j}^{(i)}(x,y)T_{j}^{(i)}(y,x))^{-1}$, with $h=1,\ldots,5$ and $i=1,\ldots,20$.

The results for the parameter estimation when applying IMTM are presented  in Table \ref{TabSV} and  show an effective improvement in the estimates, both for weekly and daily data, when compared to the results of a MH algorithm with an equivalent computational load.
\begin{table}[h]
\setlength{\tabcolsep}{4pt}
\centering
\medskip
\begin{tabular}{|c|c|cc|c|c|cc|}
\hline
&\multicolumn{3}{|c|}{Daily Data}                                         && \multicolumn{3}{|c|}{Weekly Data}\\
\hline
\hline
$\theta$        &  \multicolumn{1}{c|}{Value} &\multicolumn{2}{c|}{MSE}   & $\theta$       & \multicolumn{1}{c|}{Value} &\multicolumn{2}{c|}{MSE}    \\
                &    &\multicolumn{1}{|c}{IMTM} &\multicolumn{1}{c|}{MH}&                &       &\multicolumn{1}{|c}{IMTM} &\multicolumn{1}{c|}{MH}\\
\hline
\hline
$\alpha$        & 0  & \multicolumn{1}{c|}{ 0.04698}    &  0.09517        & $\alpha$        & 0    & \multicolumn{1}{c|}{ 0.00146   } &  0.00849  \\
                &    & \multicolumn{1}{c|}{(0.00612)}   & (0.00194)       &                 &      & \multicolumn{1}{c|}{(0.00139)  } & (0.00105) \\
\hline
$\phi$          &0.99& \multicolumn{1}{c|}{ 0.20109}   &   0.34825        & $\phi$          & 0.9  & \multicolumn{1}{c|}{ 0.01328   } &  0.10746  \\
                &    & \multicolumn{1}{c|}{(0.02414)}   & (0.05187)       &                 &      & \multicolumn{1}{c|}{(0.04014)  } & (0.03629) \\
\hline
$\sigma^{2}$    &0.01& \multicolumn{1}{c|}{ 0.00718 }   &  0.02380        & $\sigma^{2}$    & 0.1  & \multicolumn{1}{c|}{ 0.00136   } &  0.09175  \\
                &    & \multicolumn{1}{c|}{(0.00173)}   & (0.00202)       &                 &      & \multicolumn{1}{c|}{(0.00141)  } & (0.00358) \\
\hline
\end{tabular}
\caption{\label{TabSV} Mean square error (MSE) and its standard deviation (in parenthesis) for the parameter estimation with IMTM and MH within Gibbs algorithms. Left panel: daily datasets. Right panel: weekly dataset.}
\end{table}

A typical output of the IMTM for some chains of the population and for the latent variables $\mathbf{h}$ is given in Fig. \ref{FigResTypSV}. Each chart shows for a given chain the estimated latent variables (dotted black line), the posterior quantiles (gray lines) and the true value of $\mathbf{h}$ (solid black line).

Figures \ref{FigResSV1} and \ref{FigResSV2} exhibit   the HPD region at the 90\% (gray areas) and the mean (black lines) of the cumulative RMSE of each algorithm for the weekly and daily data, respectively. The statistics have been estimated from 20 independent experiments. The average RMSE shows that, in both parameter settings considered here, the IMTM (dashed black line) is more efficient than the standard MH algorithm (solid black line).

\begin{figure}[p]
\centering
\includegraphics[width=300pt,height=220pt]{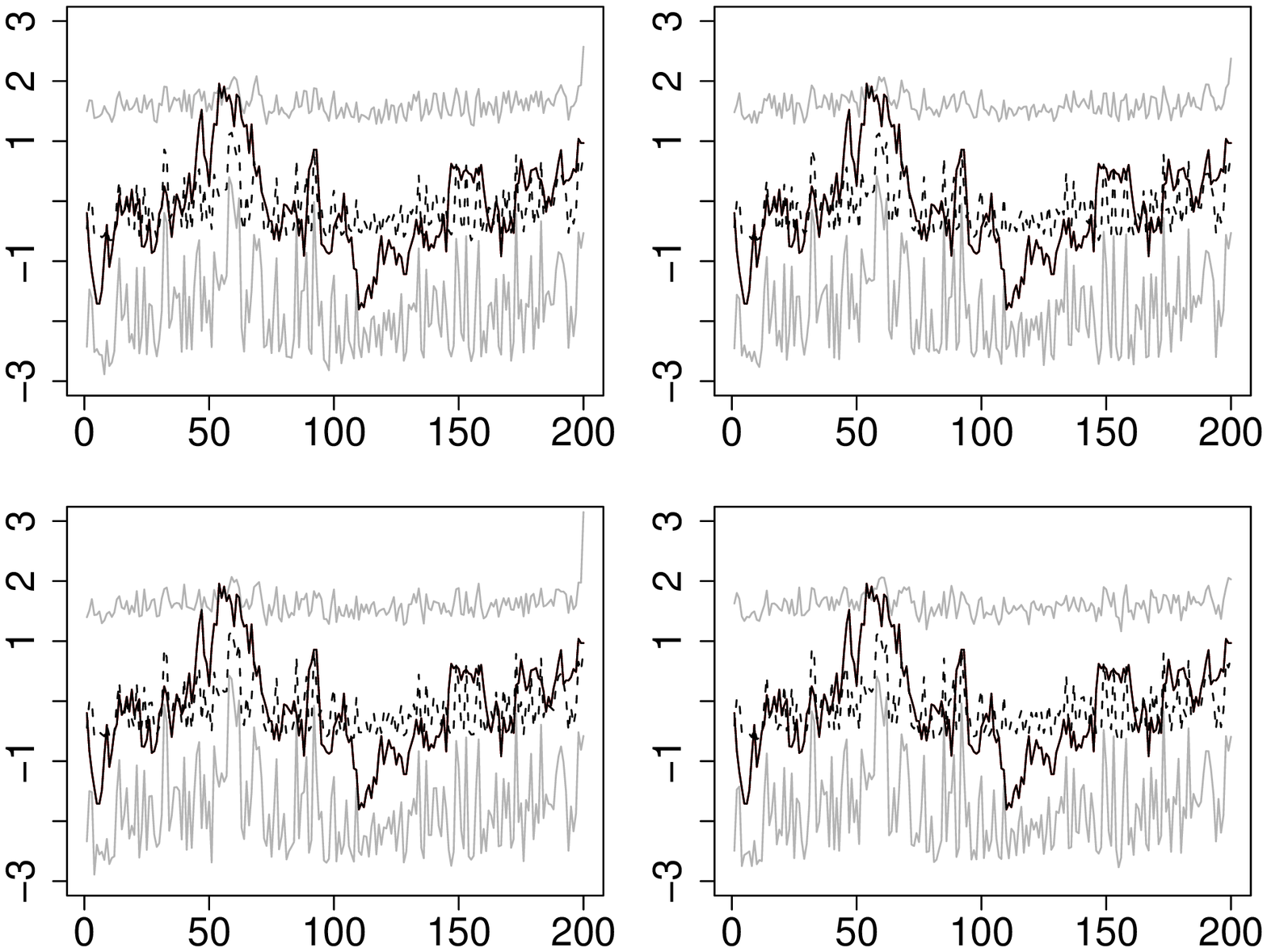}
\includegraphics[width=300pt,height=220pt]{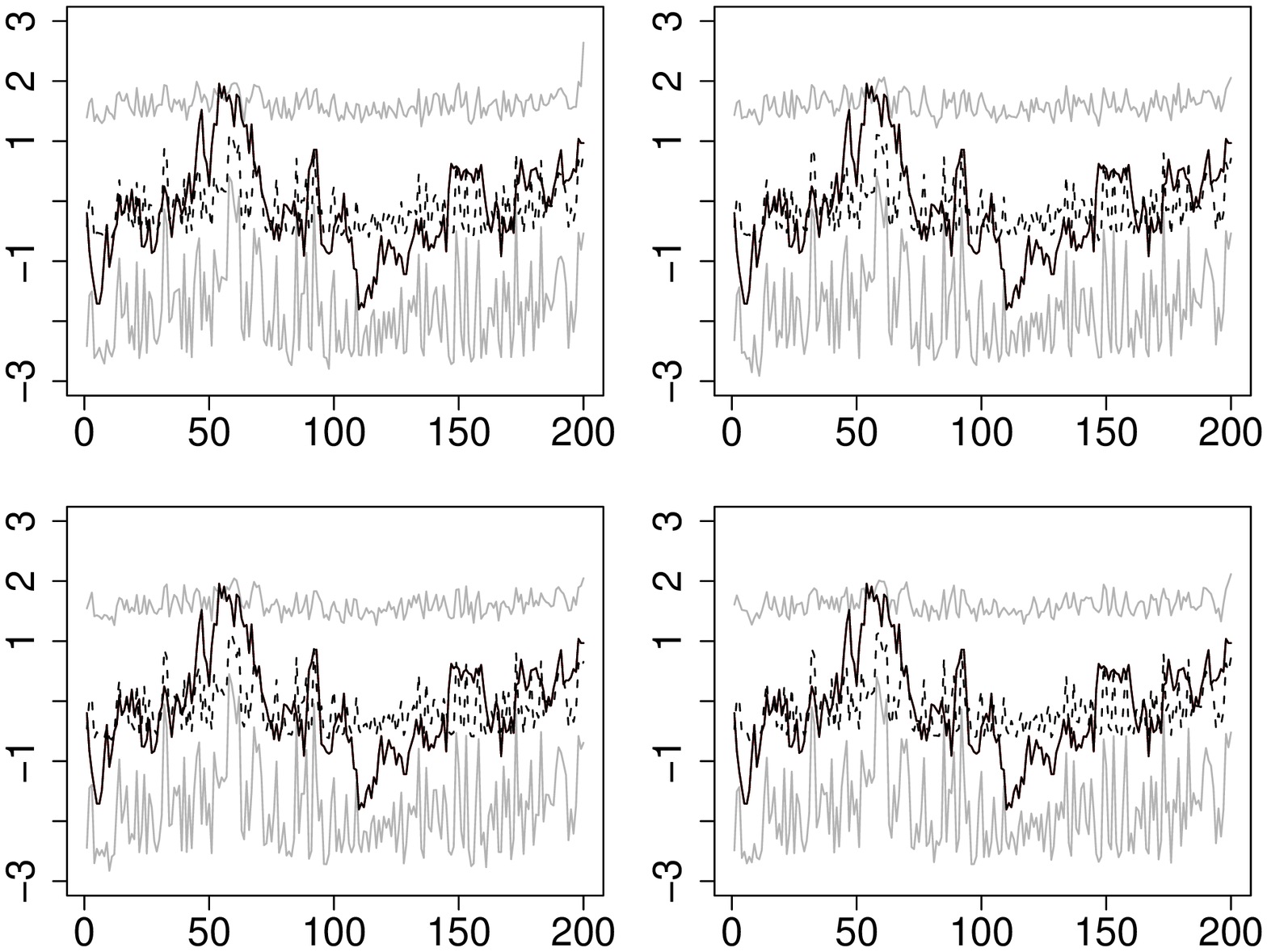}\label{FigResTypSV}
\caption{Typical output of the population of chains for a weekly dataset. For each chain the estimated latent variable (dotted black line), the 2.5\% and 97.5\% quantiles (gray lines) and the true value of $\mathbf{h}$ (solid black line).}
\end{figure}

\begin{figure}[p]
\centering
\includegraphics[width=300pt,height=220pt]{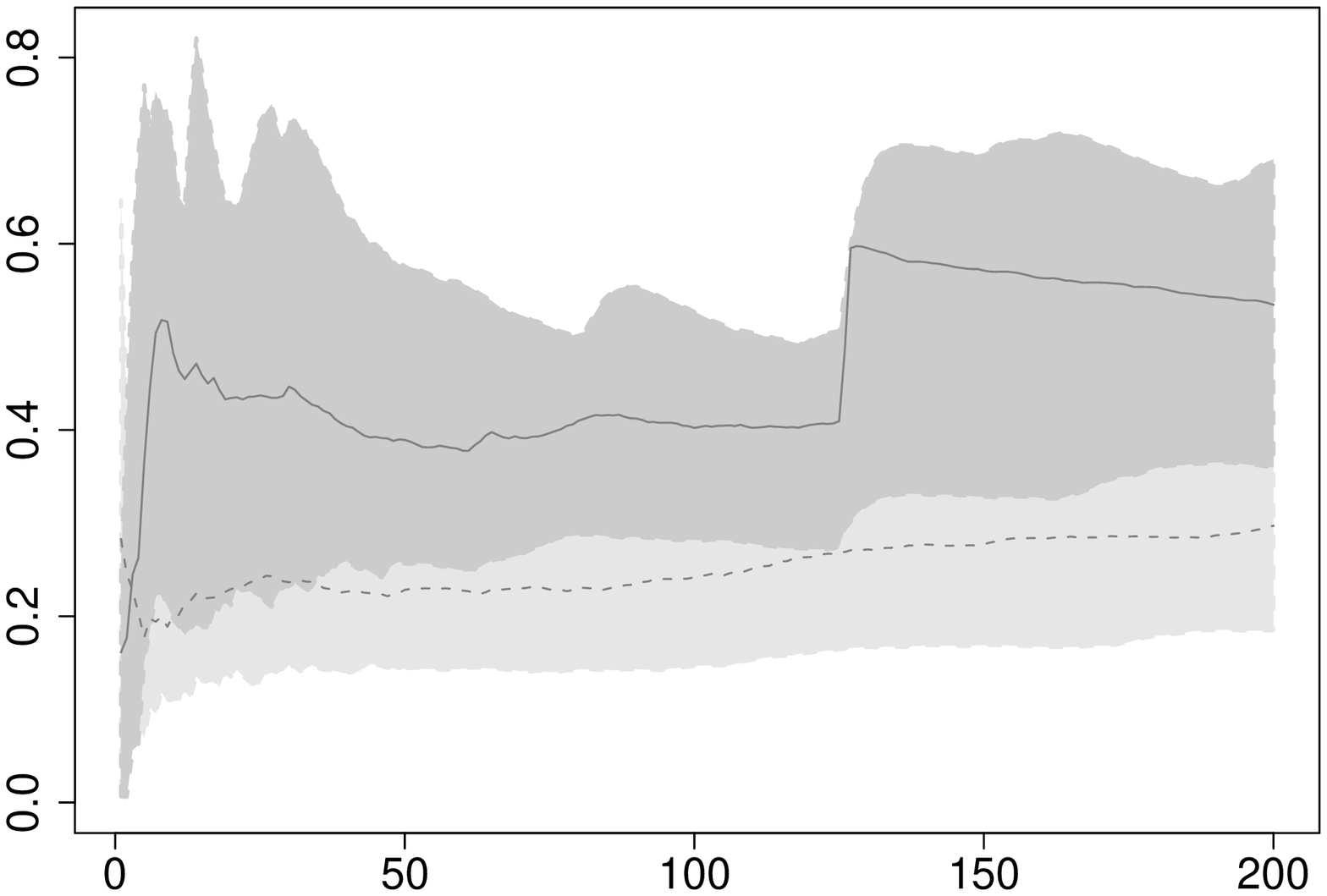}\label{FigResSV1}
\includegraphics[width=300pt,height=220pt]{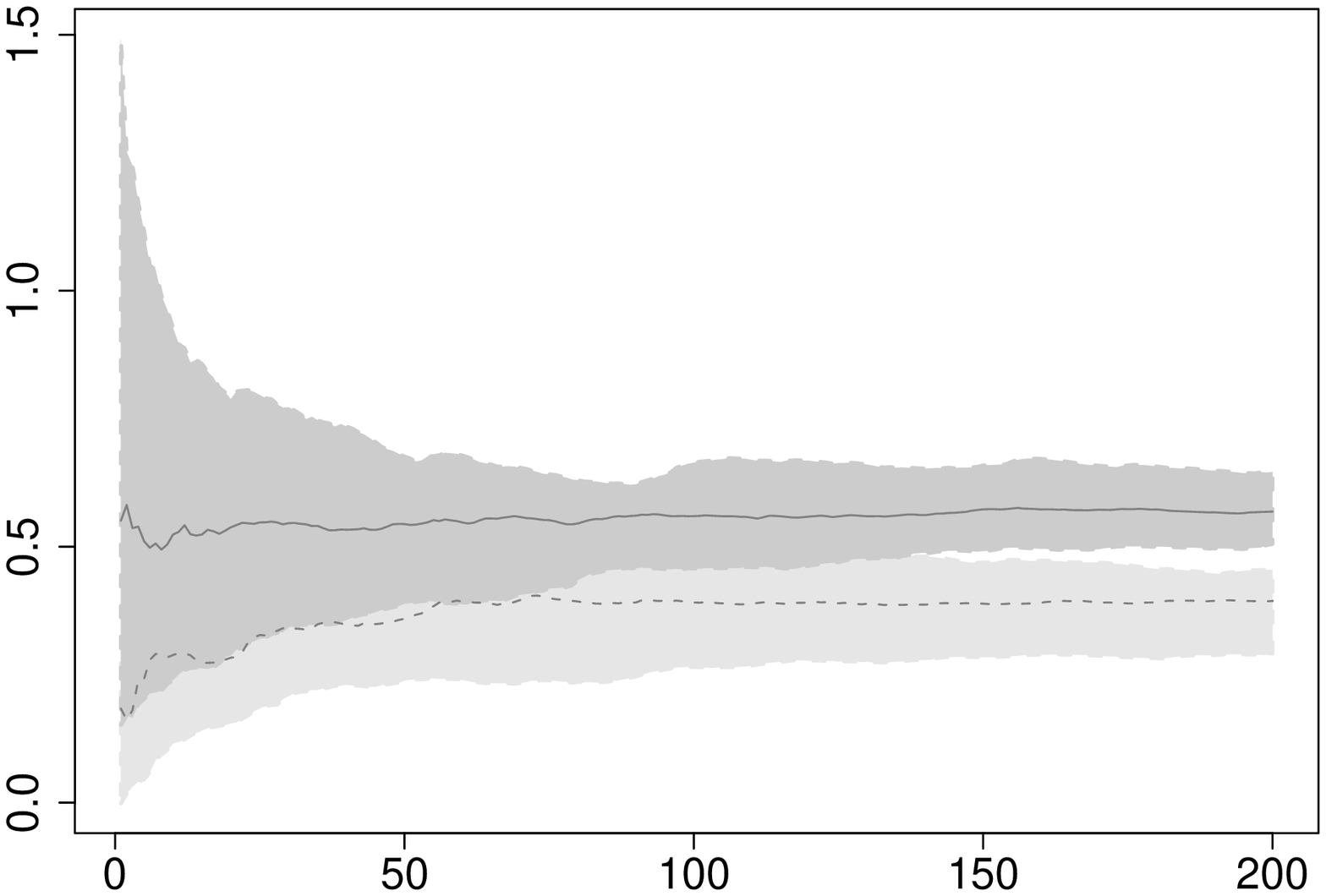}\label{FigResSV2}
\caption{Cumulative RMSE for the IMTM (dashed line) and MH (solid line) and the 90\% HPD regions for the IMTM (light gray) and MH (dark gray)
estimated on 20 independent experiments for both the daily (up) and weekly (bottom) datasets.}
\end{figure}

\section{Conclusions}
In this paper we propose a new class of interacting Metropolis algorithm,
with Multiple Try transition. These algorithms extend the existing literature
in two directions. First we show a natural and not straightforward way to include the
chains interaction in a multiple try transition. Secondly the multiple try transition has been
extended in order to account for the use of different proposal functions. We give a proof of validity
of the algorithm and show on real and simulated examples the effective improvement in the mixing property
and exploration ability of the resulting interacting chains.

\section*{Appendix A}
\begin{proof}\\
Without loss of generality, we drop the chain index $i$ and the iteration index $n$, set $M_{i}=N$, $\forall i$ and $x_n^{(i)}=x$ and denote with $T_{j}(y,x)=$ $T_{j}(y|x_{n+1}^{(1:i-1)},x_{n}^{(i)}, x_{n}^{(i+1:N)})$ $\lambda_{j}^{(i)}(y_j,x_n^{(i)})$ the $j$ proposal of the $i$-th chain at the iteration $n$ conditional on the past and current values, $x_{n+1}^{(1:i-1)}$ and $x_{n}^{(i+1:N)}$ respectively, of the other chains.

Let us define the following quantities $$\bar{w}(y_{1:N}|x)=\sum_{j=1}^{N}w_{j}(y_{j},x),\quad \bar{w}_{-k}(y_{1:N}|x)=\sum_{j\neq k}^{N}w_{j}(y_{j},x)$$
and $$S_{N}(J)=\frac{1}{\bar{w}(y_{1:N}|x)}\sum_{j=1}^{N}\delta_{j}(J)w_{j}(y_{j},x)$$
with $J\in\mathcal{J}=\{1,\ldots,N\}$ the empirical measure generated by different proposals and by the normalized selection weights.

Let $T(x,dy_{1:N})=\bigotimes_{j=1}^{N}T_{j}(x,dy_{j}):(\mathcal{Y}\times\mathcal{B}(\mathcal{Y}^{N}))\mapsto[0,1]$ the joint proposal for the multiple try and define $T_{-k}(x,dy_{1:N})=\bigotimes_{j\neq k}^{N}T_{j}(x,dy_{j})$. Let $A(x,y)$ be the actual transition probability for moving from x to y in the IMT2 algorithm. Suppose that $x\neq y$, then the transition is a results two steps. The first step is a selection step which can be written as $y=y_{J}$
and $x_{J}^{*}=x$ with the random index $J$ sampled from the empirical measure $S_{N}(J)$. The second step is a accept/reject step based on the generalized MH ratio which involves the generation of the auxiliary values $x_{j}^{*}$ for $j\neq J$. Then
\begin{eqnarray*}
&&\pi(x)A(x,y) =\\
&&= \pi(x)\int_{\mathcal{Y}^{N}}T(x,dy_{1:N})\int_{\mathcal{J}}S_{N}(dJ)\int_{\mathcal{Y}^{N-1}\times\mathcal{Y}^{2}}T_{-J}(y,dx^{*}_{1:N})\times\\
&&\quad\times\delta_{x}(dx^{*}_{J})\delta_{y_{J}}(dy) \min\left\{1,\frac{\bar{w}(y_{1:N}|x)}{\bar{w}(x_{1:N}^{*}|y)}\right\}\\
&&=\pi(x)\sum_{j=1}^{N}\int_{\mathcal{Y}^{N-1}}T_{-j}(x,dy_{1:N})T_{j}(x,y)\int_{\mathcal{Y}^{N-1}}T_{-j}(y,dx^{*}_{1:N})\times\\
&&\quad\times \frac{w_{j}(y,x)}{w_{j}(y,x)+\bar{w}_{-j}(y_{1:N}|x)}\min\left\{1,\frac{w_{j}(y,x)+\bar{w}_{-j}(y_{1:N}|x)}{w_{j}(x,y)+\bar{w}_{-j}(x_{1:N}^{*}|y)}\right\}\\
&&=\sum_{j=1}^{N}\frac{w_{j}(x,y)w_{j}(y,x)}{\lambda_{j}(y,x)}\int_{\mathcal{Y}^{2(N-1)}}T_{-j}(x,dy_{1:N})\times\\
&&\quad\times T_{-j}(y,dx^{*}_{1:N})\min\left\{\frac{1}{w_{j}(y,x)+\bar{w}_{-j}(y_{1:N}|x)},\frac{1}{w_{j}(x,y)+\bar{w}_{-j}(x_{1:N}^{*}|y)}\right\}
\end{eqnarray*}
which is symmetric in x and y.
\end{proof}


\begin{thebibliography}{30}
\expandafter\ifx\csname natexlab\endcsname\relax\def\natexlab#1{#1}\fi
\expandafter\ifx\csname url\endcsname\relax
  \def\url#1{\texttt{#1}}\fi
\expandafter\ifx\csname urlprefix\endcsname\relax\def\urlprefix{URL }\fi


\bibitem[{Barrett et~al.(1996)Barrett, Galipeau, Sanchez, Emond and
  Reid}]{barret06}
\textsc{Barrett, M.}, \textsc{Galipeau, P.}, \textsc{Sanchez, C.},
  \textsc{Emond, M.} and \textsc{Reid, B.} (1996).
\newblock Determination of the frequency of loss of heterozygosity in
  esophageal adeno-carcinoma nu cell sorting, whole genome amplification and
  microsatellite polymorphisms.
\newblock \textit{Oncogene} \textbf{12}.

\bibitem[{Campillo et~al.(2009)Campillo, Rakotozafy and
  Rossi}]{Campillo:Rakotozafy:Rossi:09}
\textsc{Campillo, F.}, \textsc{Rakotozafy, R.} and \textsc{Rossi, V.} (2009).
\newblock Parallel and interacting {M}arkov chain {M}onte {C}arlo algorithm.
\newblock \textit{Mathematics and Computers in Simulation} \textbf{79}
  3424--3433.

\bibitem[{Capp\'e et~al.(2004)Capp\'e, Gullin, Marin and Robert}]{pop-mcmc}
\textsc{Capp\'e, O.}, \textsc{Gullin, A.}, \textsc{Marin, J.} and
  \textsc{Robert, C.~P.} (2004).
\newblock Population {M}onte {C}arlo.
\newblock \textit{J. Comput. Graph. Statist.} \textbf{13} 907--927.

\bibitem[{Casarin and Marin(2009)}]{Casarin:Marin:09}
\textsc{Casarin, R.} and \textsc{Marin, J.-M.} (2009).
\newblock Online data processing: {C}omparison of {B}ayesian regularized
  particle filters.
\newblock \textit{Electronic Journal of Statistics} \textbf{3} 239--258.

\bibitem[{Casarin et~al.(2009)Casarin, Marin and
  Robert}]{Casarin:Marin:Robert:09}
\textsc{Casarin, R.}, \textsc{Marin, J.-M.} and \textsc{Robert, C.} (2009).
\newblock A discussion on: {A}pproximate {B}ayesian inference for latent
  {G}aussian models by using integrated nested {L}aplace approximations by
  {R}ue, {H}. {M}artino, {S}. and {C}hopin, {N}.
\newblock \textit{Journal of the Royal Statistical Society Ser. B} \textbf{71}
  360--362.

\bibitem[{Celeux et~al.(2006)Celeux, Marin and Robert}]{Celeux:Marin:Robert:06}
\textsc{Celeux, G.}, \textsc{Marin, J.-M.} and \textsc{Robert, C.} (2006).
\newblock Iterated importance sampling in missing data problems.
\newblock \textit{Computational Statistics and Data Analysis} \textbf{50}
  3386--3404.

\bibitem[{Chauveau and Vandekerkhove(2002)}]{Chauveau2002}
\textsc{Chauveau, D.} and \textsc{Vandekerkhove, P.} (2002).
\newblock Improving convergence of the hastings-metropolis algorithm with an
  adaptive proposal.
\newblock \textit{Scandinavian Journal of Statistics} \textbf{29} 13.

\bibitem[{Craiu and Lemieux(2007)}]{Craiu:2007rr}
\textsc{Craiu, R.~V.} and \textsc{Lemieux, C.} (2007).
\newblock Acceleration of the multiple-try {M}etropolis algorithm using
  antithetic and stratified sampling.
\newblock \textit{Statistics and Computing} \textbf{17} 109--120.

\bibitem[{Craiu and Meng(2005)}]{crameng2}
\textsc{Craiu, R.~V.} and \textsc{Meng, X.~L.} (2005).
\newblock Multi-{p}rocess {p}arallel {a}ntithetic coupling for {f}orward and
  {b}ackward {M}{C}{M}{C}.
\newblock \textit{Ann. Statist.} \textbf{33} 661--697.

\bibitem[{Craiu et~al.(2009)Craiu, Rosenthal and Yang}]{craiu-jeff-yang}
\textsc{Craiu, R.~V.}, \textsc{Rosenthal, J.~S.} and \textsc{Yang, C.} (2009).
\newblock Learn from thy neighbor: {P}arallel-chain adaptive and regional
  {MCMC}.
\newblock \textit{Journal of the American Statistical Association} \textbf{104}
  1454--1466.

\bibitem[{Del~Moral(2004)}]{DelMoral:04}
\textsc{Del~Moral, P.} (2004).
\newblock \textit{Feynman-Kac Formulae. Genealogical and Interacting Particle
  Systems with Applications}.
\newblock Springer.

\bibitem[{Del~Moral and Miclo(2000)}]{DelMoral:Miclo:00}
\textsc{Del~Moral, P.} and \textsc{Miclo, L.} (2000).
\newblock Branching and interacting particle systems approximations of
  feynmanc-kac formulae with applications to non linear filtering.
\newblock In \textit{S\'{e}minaire de Probabilit\'{e}s XXXIV. Lecture Notes in
  Mathematics, No. 1729.} Springer, 1--145.

\bibitem[{Desai(2000)}]{desai}
\textsc{Desai, M.} (2000).
\newblock \textit{Mixture Models for Genetic changes in cancer cells}.
\newblock Ph.D. thesis, University of Washington.

\bibitem[{Gelman and Rubin(1992)}]{gelmrubin}
\textsc{Gelman, A.} and \textsc{Rubin, D.~B.} (1992).
\newblock Inference from iterative simulation using multiple sequences (with
  discussion).
\newblock \textit{Statist. Sci.}  457--511.

\bibitem[{Geyer and Thompson(1994)}]{gey-temp}
\textsc{Geyer, C.~J.} and \textsc{Thompson, E.~A.} (1994).
\newblock Annealing {M}arkov chain {M}onte {C}arlo with applications to
  ancestral inference.
\newblock Tech. Rep. 589, University of Minnesota.

\bibitem[{{H}astings(1970)}]{Hast:70:MCS}
\textsc{{H}astings, W.~K.} (1970).
\newblock {Monte} {Carlo} sampling methods using {{M}arkov} chains and their
  applications.
\newblock \textit{Biometrika} \textbf{57} 97--109.

\bibitem[{Heard et~al.(2006)Heard, Holmes and Stephens}]{hhs}
\textsc{Heard, N.~A.}, \textsc{Holmes, C.} and \textsc{Stephens, D.} (2006).
\newblock A quantitative study of gene regulation involved in the immune
  response od anophelinemosquitoes: an application of {B}ayesian hierarchical
  clustering of curves.
\newblock \textit{J. Amer. Statist. Assoc.} \textbf{101} 18--29.

\bibitem[{Jasra et~al.(2007)Jasra, Stephens and Holmes}]{pop-static}
\textsc{Jasra, A.}, \textsc{Stephens, D.} and \textsc{Holmes, C.} (2007).
\newblock On population-based simulation for static inference.
\newblock \textit{Statist. Comput.} \textbf{17} 263--279.

\bibitem[{Jennison(1993)}]{jenn}
\textsc{Jennison, C.} (1993).
\newblock Discussion of "{B}ayesian computation via the {G}ibbs sampler and
  related {M}arkov chain {M}onte {C}arlo methods," by {A.F.M.} {S}mith and
  {G.O.} {R}oberts.
\newblock \textit{J. Roy. Statist. Soc. Ser. B} \textbf{55} 54--56.

\bibitem[{Liu et~al.(2000)Liu, Liang and Wong}]{Liu:2000cr}
\textsc{Liu, J.}, \textsc{Liang, F.} and \textsc{Wong, W.} (2000).
\newblock The multiple-try method and local optimization in {M}etropolis
  sampling.
\newblock \textit{Journal of the American Statistical Association} \textbf{95}
  121--134.

\bibitem[{Marinari and Parisi(1992)}]{mapa}
\textsc{Marinari, E.} and \textsc{Parisi, G.} (1992).
\newblock Simulated tempering: A new {M}onte {C}arlo scheme.
\newblock \textit{Europhysics Letters} \textbf{19} 451--458.

\bibitem[{Mengersen and Robert(2003)}]{Mengersen:Robert:03}
\textsc{Mengersen, K.} and \textsc{Robert, C.} (2003).
\newblock The pinball sampler.
\newblock In \textit{Bayesian Statistics 7} (J.~Bernardo, A.~Dawid, J.~Berger
  and M.~West, eds.). Springer-Verlag.

\bibitem[{{M}etropolis et~al.(1953){M}etropolis, Rosenbluth, Rosenbluth, Teller
  and Teller}]{metrop}
\textsc{{M}etropolis, N.}, \textsc{Rosenbluth, A.}, \textsc{Rosenbluth, M.},
  \textsc{Teller, A.} and \textsc{Teller, E.} (1953).
\newblock Equations of state calculations by fast computing machines.
\newblock \textit{J. Chem. Ph.} \textbf{21} 1087--1092.

\bibitem[{Neal(1994)}]{neal-temp}
\textsc{Neal, R.~M.} (1994).
\newblock Sampling from multimodal distributions using tempered transitions.
\newblock Tech. Rep. 9421, University of Toronto.

\bibitem[{Pandolfi et~al.(2010{\natexlab{a}})Pandolfi, Bartolucci and
  Friel}]{Pandolfi:2010b}
\textsc{Pandolfi, S.}, \textsc{Bartolucci, F.} and \textsc{Friel, N.}
  (2010{\natexlab{a}}).
\newblock A generalization of the multiple-try metropolis algorithm for
  bayesian estimation and model selection.
\newblock In \textit{13th International Conference on Artificial Intelligence
  and Statistics (AIS-TATS), Chia Laguna Resort, Sardinia, Italy}. ???

\bibitem[{Pandolfi et~al.(2010{\natexlab{b}})Pandolfi, Bartolucci and
  Friel}]{Pandolfi:2010a}
\textsc{Pandolfi, S.}, \textsc{Bartolucci, F.} and \textsc{Friel, N.}
  (2010{\natexlab{b}}).
\newblock A generalized {M}ultiple-try {M}etropolis version of the {R}eversible
  {J}ump algorithm.
\newblock Tech. rep., http://arxiv.org/pdf/1006.0621.

\bibitem[{Pritchard et~al.(2000)Pritchard, Stephens and
  Donnelly}]{Pritchard:2000fk}
\textsc{Pritchard, J.~K.}, \textsc{Stephens, M.} and \textsc{Donnelly, P.}
  (2000).
\newblock Inference of population structure using multilocus genotype data.
\newblock \textit{Genetics} \textbf{155} 945--959.

\bibitem[{Roberts and Rosenthal(2007)}]{MR2340211}
\textsc{Roberts, G.~O.} and \textsc{Rosenthal, J.~S.} (2007).
\newblock Coupling and ergodicity of adaptive {M}arkov chain {M}onte {C}arlo
  algorithms.
\newblock \textit{J. Appl. Probab.} \textbf{44} 458--475.

\bibitem[{Taylor(1994)}]{Taylor:94}
\textsc{Taylor, S.} (1994).
\newblock Modelling stochastic volatility.
\newblock \textit{Mathematical Finance} \textbf{4} 183--204.

\bibitem[{Warnes(2001)}]{Warnes:01}
\textsc{Warnes, G.} (2001).
\newblock The {N}ormal kernel coupler: {A}n adaptive {M}arkov chain {M}onte
  {C}arlo method for efficiently sampling from multi-modal distributions.
\newblock Technical report, George Washington University.

\end{thebibliography}
\end{document}